\documentclass[11pt]{article}

\usepackage{amsmath}
\usepackage{amssymb}
\usepackage{amsfonts}
\usepackage{amsthm} 
\usepackage{url}
\usepackage{comment}
\usepackage{caption}
\usepackage{graphicx}

\def\hat{\widehat}
\def\phi{\varphi}

\newtheorem{theorem}{Theorem}
\newtheorem{lemma}[theorem]{Lemma}

\newtheorem{corollary}[theorem]{Corollary}

\newtheorem{definition}[theorem]{Definition}

\renewcommand{\a}{a}
\renewcommand{\b}{b}
\newcommand{\la}{<_a}
\newcommand{\lb}{<_b}

\title{ Extension of Additive Valuations to General Valuations on the Existence of EFX}
\author{
Ryoga Mahara\thanks{Research Institute for Mathematical Sciences, Kyoto University, Kyoto, 606-8502, Japan.
E-mail: ryoga@kurims.kyoto-u.ac.jp}
}
\date{\today}

\begin{document}

\maketitle

\begin{abstract}
Envy-freeness is one of the most widely studied notions in fair division.
Since envy-free allocations do not always exist when items are indivisible, several relaxations have been considered.
Among them, possibly the most compelling concept is envy-freeness up to any item (EFX).
We study the existence of EFX allocations for general valuations.
The existence of EFX allocations is a major open problem.
For general valuations, it is known that an EFX allocation always exists (i) when $n=2$ or (ii) when all agents have identical valuations, where $n$ is the number of agents.
it is also known that an EFX allocation always exists when one can leave at most $n-1$ items unallocated.

We develop new techniques and extend some results of additive valuations to general valuations on the existence of EFX allocations.
We show that an EFX allocation always exists (i) when all agents have one of two general valuations or (ii) when the number of items is at most $n+3$.
We also show that an EFX allocation always exists when one can leave at most $n-2$ items unallocated.
In addition to the positive results, we construct an instance with $n=3$ in which an existing approach does not work as it is.

\end{abstract}

%%%%%%%%%%%%%%introduction%%%%%%%%%%%%%%%%%%%%
\section{Introduction} \label{sec: intro}
Fair division of items among competing agents is a fundamental and well-studied problem in Economics and
Computer Science. 
We are given a set $M$ of $m$ items and a set $N$ of $n$ agents with individual preferences.
Each agent $i$ has a valuation function $v_i: 2^{M} \rightarrow \mathbb{R}_{\ge 0}$ for each subset of items.
The goal is to distribute items among $n$ agents in a {\it fair} and {\it efficient} manner.
In this paper, we consider the {\it indivisible} setting: 
an item cannot be split among multiple agents.
Let an allocation $X=(X_1,X_2,\ldots, X_n)$ denote a partition of $M$ into $n$ bundles such that $X_i$ is allocated to agent $i$.
Several concepts of fairness have been considered in the literature, 
and one of the most well-studied notions of fairness is {\it envy-freeness}.
An allocation $X$ is {\it envy-free} if for any pair of agents $i, j$ we have $v_i(X_i)\ge v_i(X_j)$, i.e., no agent $i$ envies another agent $j$'s bundle.
Unfortunately, envy-free allocations do not always exist when items are indivisible.
We can easily see this even with two players and a single item having positive value for both of them: 
one of the agents has to receive the item and the other agent envies her.
This motivates the study of relaxations of envy-freeness.

The most compelling relaxations of envy-freeness is {\it envy-freeness up to any item (EFX)} introduced by Caragiannis et al.~\cite{caragiannis2019unreasonable}.
An allocation $X$ is {\it EFX} if for any pair of agents $i, j$ and for any $g\in X_j$ we have $v_i(X_i)\ge v_i(X_j\setminus \{g\})$, i.e., no agent $i$ envies another agent $j$ after the removal of {\it any} item in $j$'s bundle.
EFX is regarded as the best analogue of envy-freeness in discrete fair division: 
Caragiannis et al.~\cite{caragiannis2019envy} remarked that {\it ``Arguably, EFX is the best fairness analog of envy-freeness for indivisible items.''}
However, the existence of EFX allocations is known only in a few cases.
As described in ~\cite{caragiannis2019unreasonable}, 
{\it ``Despite significant effort, we were not able to settle the question of whether an EFX allocation always exists (assuming all goods must be allocated), and leave it as an enigmatic open question.''}

For general valuations, i.e., each valuation function $v_i$ is only assumed to be normalized and monotone, 
Plaut and Roughgarden~\cite{plaut2020almost} showed that an EFX allocation always exists (i) when $n=2$ or (ii) when all agents have identical valuations.
Furthermore, it was shown in~\cite{plaut2020almost} that exponentially many value queries may be required to identify EFX allocations even in the case where there are only two agents with identical {\it submodular} valuation functions.
It is not known whether EFX allocations always exist even when $n=3$ for general valuations.

For additive valuations, i.e., each valuation function $v_i$ is normalized, monotone, and {\it additive}\footnote{A valuation function $v_i$ is additive if $v_i(S)=\sum_{g\in S} v_i(\{g\})$ for any $S\subseteq M$.}, 
it is known that an EFX allocation always exists when $n=3$~\cite{chaudhury2020efx} or all agents have one of two valuations~\cite{mahara2020existence}.
It is not known whether EFX allocations always exist even when $n = 4$ for additive valuations.

One of relaxations of EFX is {\it EFX with charity} introduced by Caragiannis et al.~\cite{caragiannis2019envy}.
This is a {\it partial} EFX allocation, where all items need not be allocated to the agents.
Thus some items may be left unallocated.
On the other hand, an allocation is said to be {\it complete} if all items are allocated.
For general valuations, Chaudhury et al.~\cite{chaudhury2020little} showed that there exist a partial EFX allocation and a set of unallocated items $U$ such that no agent envies $U$, and $|U|\le n-1$.
For additive valuations, Caragiannis et al.~\cite{caragiannis2019envy} showed that there exists a partial EFX allocation where every agent receives at least half the value of her bundle in an optimal {\it Nash social welfare} allocation\footnote{This is an allocation that maximizes $\Pi_{i = 1}^n v_i(X_i)$.}.
Quite recently, Berger et al.~\cite{berger2021almost} showed that when $n=4$, there exists an EFX allocation with at most one unallocated item such that no agent envies the unallocated item.
Moreover, they extend their results and existing results in \cite{chaudhury2020efx} and \cite{mahara2020existence} beyond additive valuations to {\it nice cancelable} valuations which is a class including additive, unit-demand, budget-additive, multiplicative valuations, and so on.
\subsection{Our Contributions}
We investigate the existence of EFX allocations for general valuations, i.e., 
the valuation of each agent is not necessarily additive.
To prove the existence of EFX, we iteratively construct an EFX allocation from an existing partial EFX allocation to advance with respect to a certain potential function.
Chaudhury et al.~\cite{chaudhury2020efx} introduced the {\it lexicographic potential function} in order to show that they could advance an allocation while keeping EFX.
We use not only the lexicographic potential function but also a new potential function, which we call {\it partition leximin potential function}.
When we construct a new EFX allocation, which is better than the previous one with respect to the potential function, 
some agent may become worse off than in the previous allocation.
The problem is that such an agent may become to envy other agents, which results in violating EFX.
Our technical contribution is to develop a new technique to avoid such situations (see Section~\ref{sec: tech}).

Using this new technique, we obtain some new results on the existence of EFX.
Our results are described below, and are summarized in Table $1$.
Our first result is for the case where each agent has one of two given valuations.
The following theorem extends the case when all agent have the identical valuations~\cite{plaut2020almost}.
\label{sec: contribution}
\begin{theorem}\label{thm: two}
There exists a complete EFX allocation when each agent has one of two general valuations.
\end{theorem}
It is known that there exists an EFX allocation when each agent has one of two additive valuations~\cite{mahara2020existence}.
Berger et al.~\cite{berger2021almost} extended this result beyond additive valuations to nice cancelable valuations.
In \cite{mahara2020existence} and \cite{berger2021almost}, to prove the result they iteratively construct a Pareto dominating (see Section~\ref{sec: approach}) EFX allocation from an existing partial EFX allocation.
However, such an approach is not likely to work for general valuations.
To prove Theorem~\ref{thm: two}, 
we introduce a new potential function ({\it partition leximin potential function}) and show that for any EFX allocation with at least one unallocated item, one can obtain a new EFX allocation that makes progress with respect to the partition leximin potential function.
This implies that there exists a complete EFX allocation.
More details are discussed in Section~\ref{sec: two}.

Our second result concerns EFX with charity.
As mentioned above, 
it is known that there exist a partial EFX allocation and a set of unallocated items $U$ such that no agent envies $U$, and $|U|\le n-1$~\cite{chaudhury2020little}.
The following theorem strengthens the bound on the number of unallocated items from $n-1$ to $n-2$.
\begin{theorem}\label{thm: n-2}
For general valuations, there exists an EFX allocation $X$ with at most $n-2$ unallocated items.
Moreover, no agent envies the set of unallocated items in $X$.
\end{theorem}
Berger et al.~\cite{berger2021almost} showed that for nice cancelable valuations, there exists an EFX allocation $X$ with at most $n-2$ unallocated items.
Theorem~\ref{thm: n-2} extends this results to general valuations.
To prove Theorem~\ref{thm: n-2}, we show that for any EFX allocation with at least $n-1$ unallocated items, one can obtain a new EFX allocation that makes progress with respect to the lexicographic potential function.
This implies that there exists an EFX allocation with at most $n-2$ unallocated items.

We also study the case with a small number of items.
For additive valuations,  Amanatidis et al.~\cite{amanatidis2020multiple} showed that when $m\le n+2$, there exists an EFX allocation.
For general valuations, to the best of our knowledge, non-trivial results are not known.
The following theorem extends the existing results in the sense that it not only increases the number of items, but also makes valuation function general instead of additive.
\begin{theorem}\label{thm: m=n+3}
For general valuations, there exists a complete EFX allocation when $m\le n+3$.
\end{theorem}
To prove Theorem~\ref{thm: m=n+3}, we also use the lexicographic potential function.

In addition to the above positive results, we study a limitation of the approach using the lexicographic potential function.
We construct an instance with $n=3$ and $m=7$ in which there exists an EFX allocation with one unallocated item such that no progress can be made with respect to the lexicographic potential function.
This shows that Theorem~\ref{thm: n-2} and Theorem~\ref{thm: m=n+3} are the best possibilities in a sence.
\begin{table}[htb]
\caption{Our positive EFX results, where $|U|$ is the number of unallocated items.}
  \begin{tabular}{|l|l|l|} \hline
     Setting&Prior results  &Our results   \\ \hline 
     EFX for one of two valuations&Additive~\cite{mahara2020existence}, Nice cancelable~\cite{berger2021almost}  &General   \\ \hline 
     EFX with charity&
 \noindent\begin{tabular}{@{}l}
 General, $|U|\le n-1$~\cite{chaudhury2020little}\\
Nice cancelable, $|U|\le n-2$~\cite{berger2021almost}
 \end{tabular}
     &General, $|U|\le n-2$ \\ \hline
     EFX for a small number of items&Additive, $m\le n+2$~\cite{amanatidis2020multiple}  &General, $m\le n+3$   \\ \hline 
  \end{tabular}
 
\end{table}
\subsection{Our Techniques}
\label{sec: tech}
We first fix a potential function $\phi$ for all allocations.
For an existing partial EFX allocation, in order to find a new EFX allocation we use the champion graph introduced in \cite{chaudhury2020efx}.
If we have a Pareto improvable cycle (see Definition~\ref{def: PI}) in the champion graph, then we can conclude that there exists an EFX allocation $Y$ Pareto dominating $X$.
That is, we have $v_i(Y_i) \ge v_i(X_i) $ for any agent $i$, and $v_j(Y_j) > v_j(X_j)$ for some agent $j$.
It would imply that $\phi(X) < \phi(Y)$.
Otherwise, it may no longer be possible to Pareto dominate $X$.
Thus, we seek an allocation $Y$ such that some agent $i$ is worse off than in $X$ and every agent other than $i$ is not worse off than in $X$, i.e., $v_i(Y_i) < v_i(X_i)$ and $v_j(Y_j) \ge v_j(X_j)$ for $j \in N\setminus \{i\}$.
We choose such an agent $i$ who is less important with respect to $\phi$. 
In order to preserve EFX, we list bundles that can appear in $Y$.
We then allocate to agent $i$ the best of the bundles in such a list, which is a key ingredient in our construction.
Since agent $i$ receives the most favorite bundles in $Y$, agent $i$ does not envy any agent in $Y$.
We can conclude that $\phi(X) < \phi(Y)$ or the structure of the campion graph in $Y$ is better than in $X$.
In the latter case, we will find a Pareto improvable cycle in the campion graph of $Y$, and obtain a new EFX allocation $Y'$ such that $\phi(X) < \phi (Y')$.
\subsection{Related Work}
Whereas fair division of divisible resources is a classical topic starting from the 1940's~\cite{Steinhaus},
fair division of indivisible items has been actively studied in recent years.
One of the most popular relaxations of envy-freeness is {\it envy-freeness up to one item (EF1)} where no agent envies another agent after the removal of {\it some} item from the other agent’s bundle.
While the existence of EFX allocations is open, it is known that there always exists an EF1 allocations for any number of agents, and it can be computed in polynomial time~\cite{lipton2004approximately}.
There are a lot of studies on EF1 and EFX~\cite{amanatidis2021maximum, caragiannis2019unreasonable, barman2018finding, plaut2020almost, bilo2018almost, caragiannis2019envy, chaudhury2020little, chaudhury2020efx, chaudhury2021improving, berger2021almost, mahara2020existence}.
Another major concept of fairness is {\it maximin share} (MMS), which was introduced by Budish~\cite{budish2011combinatorial}.
It was shown in~\cite{kurokawa2018fair} that MMS allocations do not always exist, and 
there have been several studies on approximate MMS allocations~\cite{budish2011combinatorial, bouveret2016characterizing, amanatidis2017approximation, barman2017approximation, kurokawa2018fair, ghodsi2018fair, garg2018approximating, garg2020improved}.
In addition, study on finding {\it efficient} fair allocations has attracted attention.
{\it Pareto-optimality} is a major notion of efficiency.
Caragiannis et al.~\cite{caragiannis2019unreasonable} showed that any allocation that has maximum Nash social welfare is guaranteed to be Pareto-optimal and EF1.
Unfortunately, finding an allocation with the maximum Nash social welfare is APX-hard~\cite{lee2017apx}.
There are several studies on approximation algorithms for maximizing Nash social welfare~\cite{cole2018approximating, cole2017convex, chaudhury2018fair, anari2017nash, anari2018nash, barman2018finding, garg2018approximating, garg2020approximating, li2021constant}.

There are many real-world scenarios where items or resources need to be divided fairly, e.g., taxi fare division, rent division, task distribution, and so on.
Spliddit (www.spliddit.org) is a fair division website, which offers a fair solution for the division of rent, goods, and credit \cite{goldman2015spliddit}.
This website implements mechanisms for users to log in, define what is to be divided, enter their valuations, and demonstrate fair division.
Since its launch in 2014, there have been several thousands of users~\cite{caragiannis2019unreasonable}.
For more details on Spliddit, we refer to the reader to~\cite{goldman2015spliddit, plaut2020almost}.
Another fair division application is {\it Course Allocation} used at the Wharton School at the University of Pennsylvania to fairly allocate courses among students~\cite{plaut2020almost, budish2017course}.
\subsection{Organization}
In Section~\ref{sec: pre}, we present the model, denote some basic notions introduced by \cite{chaudhury2020efx, chaudhury2020little}, and prove some useful lemmas.
In Section~\ref{sec: n-2}, we consider EFX with charity for general valuations, and prove Theorem~\ref{thm: n-2}.
In Section~\ref{sec: two}, we consider the setting with only one of two types of general valuations, and prove Theorem~\ref{thm: two}.
In Section~\ref{sec: m=n+3}, we study on the setting with a small number of items, and prove Theorem~\ref{thm: m=n+3}.
In Section~\ref{sec: limit}, we construct an instance with $n=3$ and $m=7$ that shows a certain limitations of the approach of the lexicographic potential function.
%In Section~\ref{sec: non-degenerate}, we prove that in order to show the existence of EFX, we can assume that an instance is non-degenerate.
%In Section~\ref{sec: m=n+3}, we study on the setting with a small number of items, and prove Theorem~\ref{thm: m=n+3}.

%%%%%%%%%%%%%%Preliminaries%%%%%%%%%%%%%%%%%%%
\section{Preliminaries}
\label{sec: pre}
Let $N=\{1,\dots,n \}$ be a set of $n$ agents and $M$ be a set of $m$ items.
In this paper, we assume that items are indivisible: an item may not be split among multiple agents.
Each agent $i \in N $ has a valuation function $v_i : 2^M \rightarrow \mathbb{R}_{\ge 0}$.
We assume that (i) any valuation function $v_i$ is {\it normalized}: $v_i(\emptyset) = 0$ and (ii) it is {\it monotone}:  $S\subseteq T$ implies $v_i(S) \le v_i(T)$ for any $S,T \subseteq M$.

To simplify notation, we denote $[k]$ by $\{1,\dots,k\}$, write $v_i(g)$ instead of $v_i(\{g\})$ for $g\in M$, and use $S\setminus g, S\cup g$ instead of $S\setminus \{g\}, S\cup \{g\}$, respectively.
We also denote $S <_i T$ instead of $v_i (S) < v_i (T)$. 
In a similar way, we use the symbols $>_i, \le_i, $ and $\ge_i$.

For $M' \subseteq M$, an {\it allocation} $X=(X_1,X_2,\ldots, X_n)$ on $M'$ is a partition of $M'$ into $n$ disjoint subsets, where $X_i$ is the {\it bundle} given to agent $i$. 
We say that an allocation $X=(X_1,X_2,\ldots, X_n)$ on $M'$ is {\it complete} if $M' = M$.
Otherwise, we say that an allocation $X$ is {\it partial}.

Given an allocation $X$, we say that agent $i$ {\it envies} a set of items $S$ if $X_i <_i S$. 
We say that agent $i$ {\it envies} agent $j$ if $i$ envies $X_j$.
We say that agent $i$ {\it EFX envies} a set of items $S$ if there exists some $h\in S$ such that $i$ envies $S\setminus h$.
We say that agent $i$ {\it EFX envies} agent $j$ if $i$ EFX envies $X_j$.
Note that if $i$ EFX envies $j$ then $i$ envies $j$, but not vice versa.
An allocation $X$ is called {\it envy-free} if no agent envies another.
An allocation $X$ is called {\it EFX} if no agent EFX envies another.

An instance $I$ is a triple $\langle N, M, \mathcal{V} \rangle $, where $\mathcal{V}=\{v_1,\dots,v_n\}$ is a set of valuation functions.
We use an assumption on instances considered in \cite{chaudhury2020efx}.
\begin{definition} An instance $I$ is {\rm non-degenerate} if for any $i \in N$ and $S,T \subseteq M$, 
\[S\neq T \Rightarrow v_i(S) \neq v_i(T).\]
\end{definition}
We can show that in order to prove the existence of an EFX allocation, we may assume w.l.o.g.~that instances are non-degenerate.
This assumption was considered for additive valuations in \cite{chaudhury2020efx}, and we can easily extend it for general valuations.
More details are presented in Appendix~\ref{sec: non-degenerate}.
%More details are presented in Appendix~\ref{sec: non-degenerate}.
In what follows, {\it we only deal with non-degenerate instances}.

\subsection{Overall Approach}\label{sec: approach}
All of our results on the existence of EFX can be viewed in a unified framework as follows:
we first fix an appropriate potential function $\phi$ on all allocations.
We then show that given any partial EFX allocation $X$, one can construct a new EFX allocation $Y$ that makes progress with respect to the potential function, i.e., $\phi(X) < \phi(Y)$.
Since there are finitely many allocations, there must exist a complete EFX allocation.
One of natural potential functions is {\it social welfare}.
Given an allocation $X$, we denote the {\it social welfare} of $X$ by $\varphi(X)=\sum_{i\in N} v_i(X_i)$. 
A typical notion to progress the social welfare is Pareto domination.
Given two allocations $X,Y$, we say that $Y$ {\it Pareto dominates} $X$ if $Y_i\ge_i X_i$ for all $i\in N$, and $Y_j >_j X_j$ for some agent $j\in N$.
Clearly, if $Y$ Pareto dominates $X$, then $\varphi(Y)>\varphi(X)$.
Chaudhury et al.~\cite{chaudhury2020efx} have shown that there does not always exist a Pareto-dominating EFX allocation when $n=3$ for additive case.
To overcome this barrier they introduce a {\it lexicographic potential function}, which we also use to prove Theorems~\ref{thm: n-2} and \ref{thm: m=n+3}.
In addition, to prove Theorem~\ref{thm: two}, we use a new potential function.
More detail on each potential function is presented in Sections~\ref{sec: n-2} and \ref{sec: two}.
\subsection{Minimum Preferred Set and Most Envious Agent}
Most envious agent is a basic notion, introduced in \cite{chaudhury2020little}.
Consider an allocation $X=(X_1,X_2,\ldots, X_n)$ and a set $S \subseteq M$.
For an agent $i$ such that $S >_i X_i$, we define a {\it minimum preferred set} $P_X(i,S)$ of agent $i$ for $S$ with respect to allocation $X$ as a smallest cardinality subset $S'$ of $S$  such that $S' >_i X_i$. 
Define $\kappa _X(i,S)$ by
\begin{align*}
  \kappa _X(i,S)=
  \left\{
    \begin{array}{ll}
      |P_X(i,S)| & {\rm if}~S >_i X_i,\\
      +\infty & {\rm otherwise}.
    \end{array}
  \right.
\end{align*}
Let $\kappa_X(S) = \min _{i \in N}  \kappa _X(i,S)$.
We define $A_X(S)$ for a set $S$ as the set of agents with the smallest values of $\kappa _X(i,S)$, i.e., 
\[A_X(S) = \{ i\in N \mid S >_i X_i ~{\rm and}~ \kappa _X(i,S)= \kappa_X(S) \}.\]
%If $\kappa_X(S) < +\infty$, 
%どうせA_X(S)の定義から，k_X(S) = \infty ならば A_X(S)は空集合となるから大丈夫
We call $A_X(S)$ the set of {\it most envious agents}.
Note that if agent $i$ is a most envious agent for $S$, it holds that $i$ envies $P_X(i,S)$ and no agent EFX envies $P_X(i,S)$.
When it is clear from the context, we abbreviate $P_X(i,S)$ as $P$.
\subsection{Champions and Champion Graph}
In order to find a new EFX allocation from the existing partial EFX allocation, 
Champions and Champion graph are important notions, introduced in \cite{chaudhury2020efx}.
Let $X$ be a partial allocation on $M'\subsetneq M$ and let $g \in M\setminus M'$ be an unallocated item.
For agents $i$ and $j$ (possibly $i=j$), we say that $i$ $g$-{\it champions} $j$ if $i$ is a most envious agent for $X_j\cup g$.
Then, we also call $i$ a $g$-{\it champion} of $j$.
When $i$ is a most envious agent for $X_i\cup g$, we call $i$ a {\it self $g$-champion}.
Note that every agent $j$ has a $g$-champion.
Indeed, since instances are non-degenerate, and valuations are monotone, we have $X_j \cup g >_j X_j$.
That is, since $j$ envies $X_j \cup g $, there exists at least one most envious agent for $X_j \cup g$.

\begin{comment}
\begin{lemma}
Assume that $i$ $g$-champion $j$.
Let $P$ be a minimum preferred set of $i$.
Then, for any item $h\in P$ and any agent $k\in N$ (possibly $k=i$), it holds that $P\setminus h \le_{k} X_k$.
\end{lemma}
\end{comment}

We say that $i$ $g$-{\it decomposes} $j$ if $i$ $g$-champions $j$, and $\{g\} \subsetneq P \subsetneq X_j\cup g$, where $P$ is a minimum preferred set of $i$ for $X_j\cup g$.
When $i$ $g$-decomposes $j$, we can decompose $X_j$ into $P\setminus g$ and $(X_j\cup g)\setminus P$.
If there is no ambiguity, then $T_j=P\setminus g$ and $B_j=(X_j\cup g)\setminus P$ are called {\it top and bottom half-bundles} of $X_j$, respectively.
The following lemma illustrates a typical situation such that $i$ $g$-decomposes $j$.
\begin{lemma}\label{lem: decompose}
If $i$ $g$-champions $j$, $i$ dose not envy $j$, and both $i$ and $j$ are not self $g$-champions, then $i$ $g$-decomposes $j$.
\end{lemma}
\begin{proof}
By the assumption, we have $X_j <_i X_i <_i X_j \cup g$.
Let $P$ be a minimum preferred set of $i$ for $X_j \cup g$.
If $g\notin P$, then $P\subseteq X_j$, and by the monotonicity, 
we have $P \le_i X_j <_i X_i$.
This contradicts the definition of $P$.
Thus, $g \in P$.
If $P=\{g\}$, then $\kappa_{X}(i, X_i\cup g)=1$, and hence it contradicts that $i$ is not a self $g$-champion.
Thus, $\{g\} \subsetneq P$.
Furthermore if $P=X_j \cup g$, then $\kappa_X(X_j \cup g)=|P|$ and hence it contradicts that $j$ is not a self $g$-champion.
Therefore $\{g\} \subsetneq P \subsetneq X_j \cup g$, and thus $i$ $g$-decomposes $j$.
\end{proof}
\begin{definition}
The {\em champion graph} $M_{X}=(N,E)$ with respect to allocation $X$ is a labeled directed multi-graph.
The vertices correspond to the agents, and $E$ consists of the following two types of edges:
\begin{enumerate}
\item Envy edges: $i \rightarrow j$ iff $i$ envies $j$.
\item Champion edges: $i \xrightarrow{g} j$ iff $i$ $g$-champions $j$, where $g$ is an unallocated item.
\end{enumerate}
\end{definition}
Envy graph which consists of only envy edges is introduced in~\cite{lipton2004approximately}.
The original champion graph considered in~\cite{chaudhury2020efx} consists of only champion edges.
Our definition of champion graph combines these two notions for convenience.
Recently, Berger et al.~\cite{berger2021almost} denote the generalized champion graph that contains more additional edges.
For convenience, $i \rightarrow j$ and $ i \xrightarrow{g} j$ are sometimes denoted by $i \xrightarrow{\emptyset} j$ and $ i \xrightarrow{\{g\}} j$, respectively.
Berger et al.~\cite{berger2021almost} denote the notion of Pareto improvable cycle, which is very useful in our argument.
\begin{definition}\label{def: PI}
A cycle $C=a_1 \xrightarrow{H_1} a_2 \xrightarrow{H_2} \cdots \xrightarrow{H_{k-1}} a_k \xrightarrow{H_k} a_1$ in $M_X$ is called {\em Pareto improvable} $(\mathbf{PI})$ if for every $i,j\in [k]$ we have $H_i \cap H_j=\emptyset$, where $H_i$ is an empty set or a singleton of an unallocated item.
\end{definition}
The following lemma shows that if we have a Pareto improvable cycle in $M_X$, then there exists an allocation $Y$ that Pareto dominates $X$ while keeping EFX.
\begin{lemma}(Berger et al.~\cite{berger2021almost})\label{lem: PI}
Let $X$ be an allocation.
If $M_X$ contains a Pareto improvable cycle, then there exists an allocation $Y$ Pareto dominating $X$ such that for any $i, j \in N$, 
if $i$ does not $EFX$ envy $j$ in $X$, then neither in $Y$.
In particular, if $X$ is an EFX allocation, then so is $Y$.
Furthermore, every agent $i$ along the cycle satisfies $X_i <_i Y_i$.
\end{lemma}

\begin{comment}
\begin{proof}
Let $C=a_1 \xrightarrow{H_1} a_2 \xrightarrow{H_2} \cdots \xrightarrow{H_{k-1}} a_k \xrightarrow{H_k} a_1$ be a Pareto-improvable cycle in $M_X$.
We define a new allocation $Y$ as follows: 
\begin{align*}
Y_i&=  P_i & ~{\rm for~any}~ i ~ {\rm on~ cycle}~ C,\\
Y_j&= X_j &{\rm otherwise,}\\
\end{align*}
where, $P_i$ is a minimum preferred set of $a_i$.
We can easily check that $Y$ is indeed an allocation.
For any $i$ on $C$, by the definition of minimum preferred set, $X_i <_i P_i=Y_i$.
Since any agents outside the cycle does not change their bundles, $Y$ Pareto-dominates $X$.

It remains show that for any $i, j \in N$, 
$i$ does not $EFX$ envy $j$ in $X$, then neither does her in $Y$.
Let $i', j' \in N$ be any two agents.
Since no agent becomes worse off, it suffices to show that no agent EFX envies $Y_i$ for any $i\in N$ in allocation $X$.
If $i$ is outside the cycle, by the fact that $X$ is EFX, no agent EFX envies $Y_i=X_i$ in $X$.
If $i$ is on the cycle, by the definition of minimum preferred set, no agent EFX envies $Y_i=P_i$ in $X$.
\end{proof}
\end{comment}

\begin{corollary}\label{cor: PI}
Let $X$ be an EFX allocation.
If $M_X$ contains an envy-cycle\footnote{envy-cycle is a dicycle in $M_X$ composed of only envy edges.}, a self $g$-champion, or a cycle composed of envy edges and at most one champion edge, then there exists an EFX allocation $Y$ that Pareto dominates $X$.
\end{corollary}

%%%%%%%%%%%%%%
\section{Existence of EFX with at most $n-2$ unallocated items}
\label{sec: n-2}
In this section, we prove Theorem~\ref{thm: n-2}.
We use a lexicographic potential function as in \cite{chaudhury2020efx}.
Recall that $N=\{1,\dots, n\}$.
For an allocation $X$, the {\it lexicographic potential function} $\phi(X)$ is defined as the vector $(v_{1}(X_{1}),\dots,v_{n}(X_{n}))$.
Intuitively, agent $1$ is the most important agent and agent $n$ is the least important agent in $N$.
\begin{definition}\label{def: lexi}
For two allocations $X,Y$, 
We denote $Y \succ_{\rm lex} X$ if $\phi(Y)$ is lexicographically larger than $\phi(X)$, i.e., for some $k\in N$, we have that $Y_{j} =_{j} X_{j}$ for all $1\le j < k$, and $Y_{k} >_{k} X_{k}$.
\end{definition}
Note that if $Y$ Pareto dominates $X$ then $Y \succ_{\rm lex} X$, but not vice versa.
The following basic lemma is shown in \cite{berger2021almost}, which we also use.
\begin{lemma}\label{lem: progress}(Berger et al.~\cite{berger2021almost})
If for every partial EFX allocation $X$ with $k$ unallocated items, there exists an EFX allocation $Y$ such that $Y \succ_{\rm lex} X$, then there exists an EFX allocation with at most $k-1$ unallocated items.
Moreover, no agent envies the set of $k-1$ unallocated items.
\end{lemma}
By Lemma~\ref{lem: progress}, in order to prove Theorem~\ref{thm: n-2}, it suffices to show that for every partial EFX allocation $X$ with at least $n-1$ unallocated items, there exists an EFX allocation $Y$ such that $Y \succ_{\rm lex} X$.
We first prove the following lemma, which is used in the proof of Theorem~\ref{thm: n-2}.
\begin{lemma}\label{lem: one edge}

Let $X$ be an EFX allocation with at least $n-1$ unallocated items.
Then, there exists an EFX allocation $Y$ Pareto dominating $X$ in the following two cases.
\begin{enumerate}
\item[$(1)$] the number of unallocated items is at least $n$.
\item[$(2)$] there exists at least one envy edge $j \rightarrow i$ in $M_X$.
\end{enumerate}
Moreover, in case $(2)$, some agent $l \in N\setminus i$ is strictly better off than in $X$, i.e., $Y_l >_l X_l$.
\end{lemma}
\begin{proof}
Let $\{g_1,\dots, g_k\}$ denote the set of unallocated items.
We first prove the case of $(1)$.
It suffices to prove that there exists a PI cycle in $M_X$ by Lemma~\ref{lem: PI}.
Let $a_1$ be an arbitrary agent.
Then, some agent $a_2$ $g_1$-champions $a_1$.
If $a_2=a_1$, then we have a PI cycle, and we are done.
Assume that $a_2\neq a_1$.
Then, some agent $a_3$ $g_2$-champions $a_2$.
If $a_3=a_1~{\rm or }~a_2$, then we have a PI cycle.
Indeed, in the first case we have a cycle $a_1  \xrightarrow{g_2} a_2  \xrightarrow{g_1}  a_1 $, and in the second case we have a self $g_2$-champion.
We can continue this way to conclude that w.l.o.g.~we have a directed path $a_n  \xrightarrow{g_{n-1}} a_{n-1}  \xrightarrow{g_{n-2}} \cdots  \xrightarrow{g_1} a_1$ in $M_X$, where $a_1,\dots,a_n$ are different agents.
Now, some agent $g_n$-champions $a_n$ in $M_X$.
No matter who it is, there exists a PI cycle, and we are done.

We prove the case of $(2)$ in a similar way.
Assume w.l.o.g that some agent $a_2$ envies $a_1$.
By a similar argument as above, we can conclude that w.l.o.g. we have a directed path
$a_n  \xrightarrow{g_{n-1}} a_{n-1}  \xrightarrow{g_{n-2}} \cdots  \xrightarrow{g_3}a_3 \xrightarrow{g_2} a_2 \rightarrow a_1$ in $M_X$, where $a_1,\dots,a_n$ are different agents.
Now, some agent $g_1$-champions $a_n$.
No matter who it is, there exists a PI cycle, and we are done.
Moreover, in any cases, we have a PI cycle containing some agent in $N\setminus a_1$.
Hence, the last statement of lemma holds by Lemma~\ref{lem: PI}.
\end{proof}
We are now ready to prove Theorem~\ref{thm: n-2}.
We fix an arbitrary ordering of the agents.
\begin{proof}[Proof of Theorem \ref{thm: n-2}]
Let $X$ be an EFX allocation with $k \ge n-1$ unallocated items, and let $\{g_1,\dots, g_k\}$ denote the set of unallocated items.
By Lemma~\ref{lem: progress}, it suffices to prove that there exists an EFX allocation $Y$ such that $Y \succ_{\rm lex} X$.
By Lemma~\ref{lem: one edge}, 
when $k\ge n$, or $k=n-1$ and there exists at least one envy edge in $M_X$, we are done.
Assume that $k=n-1$ and there exists no envy edge in $M_X$.
Let $a_1$ be the last agent in the fixed ordering, i.e., $a_1$ is the least important agent in the lexicographic potential function.
By a similar argument in Lemma~\ref{lem: one edge}, we conclude that w.l.o.g. we have a directed path
$a_n  \xrightarrow{g_{n-1}} a_{n-1}  \xrightarrow{g_{n-2}} \cdots  \xrightarrow{g_1} a_1$ in $M_X$, where $a_1,\dots,a_n$ are different agents.
Furthermore, we may assume that there are no self $g_i$-champions for $1\le i \le n-1$ since otherwise we have a PI-cycle.
Since there are no self-champions, and there exists no envy edge in $M_X$, $a_{i+1}$ $g_i$-decomposes $a_i$ for $1\le i \le n-1$ by Lemma~\ref{lem: decompose}.
Let $T_i$ and $B_i$ are the top and bottom half-bundles of $X_{a_i}$ decomposed by $a_{i+1}$ for $1\le i \le n-1$, respectively.
Consider $Z=\max_{a_1}\{T_1\cup g_1, T_2\cup g_2,\dots,T_{n-1}\cup g_{n-1}, X_{a_2}, \dots, X_{a_n}\}$\footnote{$\max_{a_1}\{T_1\cup g_1, T_2\cup g_2,\dots,T_{n-1}\cup g_{n-1}, X_{a_2}, \dots, X_{a_n}\}$ is $a_1$'s most favorite bundle out of $T_1\cup g_1, T_2\cup g_2,\dots,T_{n-1}\cup g_{n-1}, X_{a_2}, \dots, X_{a_n}$.}.
We discuss in two cases.
\begin{description}
\item[Case 1:]  $Z=T_i \cup g_i$ or $X_{a_i}$ for $2\le i \le n$

We define a new allocation $X'$ as follows:
\begin{align*}
X'_{a_1} &= Z,&\\
X'_{a_j} &= T_{j-1}\cup g_{j-1} & {\rm for}\ 2 \le j \le i, \\
X'_{a_j} &= X_{a_j} & {\rm for}\ i < j \le n.
\end{align*}
We show that $X'$ is EFX and $X' \succ_{\rm lex} X$.
For $1\le t \le i-1$, since $T_t\cup g_t$ is a minimum preferred set of $a_{t+1}$ and $a_{t+1}$ is a most envious agent for $T_t\cup g_t$, no agent EFX envies $T_t\cup g_t$ in $X$.
Thus, for $2\le s \le n$, since $X'_{a_s} \ge_{a_s} X_{a_s}$, agent $a_s$ does not EFX envy $T_t\cup g_t$ in $X'$.
For $2\le s \le n$, since $X$ is envy-free and the fact that $X'_{a_s} \ge_{a_s} X_{a_s}$, agent $a_s$ does not envy $X_{a_u}$ for $1\le u \le n$ in $X'$.
By the definition of $Z$, $a_1$ does not envy any agents in $X'$.
Therefore, $X'$ is EFX.
Furthermore, for $2\le j \le i$, each agent $a_j$ is strictly better off than in $X$, and each agent $a_j$ does not change her bundle for $i< j \le n$.
Thus, we have $X' \succ_{\rm lex} X$, and we are done.

\item[Case 2:]  $Z=T_1\cup g_1$

We define a new allocation $X'$ as follows:
\begin{align*}
X'_{a_1} &= Z,&\\
X'_{a_i} &= X_{a_i} & {\rm for}\ 2 \le i \le n. 
\end{align*}
We show that $X'$ is EFX.
Since we change only $a_1$'s bundle from $X$, it is enough to check that there is no EFX envy from or to $a_1$.
By the definition of $Z$, $a_1$ does not envy any agent in $X'$.
Since $Z = T_1\cup g_1$ is a minimum preferred set of $a_2$ for $X_{a_1}\cup g_1$, and $a_2$ is a most envious agent for $X_{a_1}\cup g_1$, no agent EFX envies $T_1\cup g_1$ in $X'$.
Thus $X'$ is EFX.
In addition, since $Z=T_1\cup g_1$ is a minimum preferred set of $a_2$, $a_2$ envies $a_1$ in $X'$.
By the fact that $B_1\neq \emptyset$, we now have at least $n-1$ items in $\{g_2,\dots, g_{n-1}\}\cup B_1$ that are unallocated.
Thus by the case of (2) in Lemma~\ref{lem: one edge}, there exists an EFX allocation $X''$ that Pareto dominates $X'$.
Furthermore, there exists some agent $a_i\ (2\le i \le n)$ such that $X''_{a_i} >_{a_i} X'_{a_i} = X_{a_i}$. Since $X''_{a_j} \ge_{a_j} X'_{a_j}=X_{a_j}$ for $2\le j \le n$, we have $X'' \succ_{\rm lex} X$, and we are done.

\end{description}
\end{proof}

\section{Existence of EFX with One of Two General Valuations}
\label{sec: two}

In this section, we prove Theorem~\ref{thm: two}.
For two general valuation functions $v_a$ and $v_b$, 
let $N_a$ (resp.~$N_b$) be the set of agents whose valuation is $v_a$ (resp.~$v_b$).
To prove Theorem~\ref{thm: two}, we introduce a new potential function.
For an allocation $X$, 
we write $N_\a = \{\a_0, \a_1, \ldots , \a_s\}$ and $N_\b = \{\b_0, \b_1, \ldots , \b_t\}$, where $X_{\a_0} \la X_{\a_1} \la \cdots \la X_{\a_s}$
and $X_{\b_0} \lb X_{\b_1} \lb \cdots \lb X_{\b_t}$.
Define the {\it partition leximin potential function} $\psi(X)$ as the vector $(v_{a}(X_{a_0}),\dots,v_{a}(X_{a_s}), v_{b}(X_{b_0}),\dots,v_{b}(X_{b_t}))$.
\begin{definition}
For two allocations $X,Y$, 
we denote $Y \succ_{\rm p.lexmin} X$ if $\psi(Y)$ is lexicographically larger than $\psi(X)$.
\end{definition}
That is, we prioritize $N_a$ over $N_b$, compare agents in $N_a$ by the leximin ordering, and second compare agents in $N_b$ by the leximin ordering.
Note that if $Y$ Pareto dominates $X$, then $Y \succ_{\rm p.lexmin} X$ but not vice versa.
Our goal is to show the following theorem.

\begin{comment}
In this section, we prove our main result, Theorem~\ref{thm: two}.
Given two general valuation functions $v_a$ and $v_b$, 
we denote $N_\a = \{\a_0, \a_1, \ldots , \a_s\}$ and $N_\b = \{\b_0, \b_1, \ldots , \b_t\}$, where $N_a$ (resp.~$N_b$) is the set of agents whose valuation is $v_a$ (resp.~$v_b$).
To prove Theorem~\ref{thm: two}, we introduce a new potential function.
For an allocation $X$, let a {\it divided leximin potential function} $\psi(X)$ be the vector $(v_{a}(X_{a'_0}),...,v_{a}(X_{a'_s}), v_{b}(X_{b'_0}),...,v_{b}(X_{b'_t}))$, where $a'_0,...,a'_s$ (resp.~$b'_0,...,b'_s$) are sorted in decreasing values of $a_0,...,a_s$ (resp.~$b_0,...,b_t$), i.e., $v_a(X_{a'_0})<_a \cdots <_a v_a(X_{a'_s})$ and $v_b(X_{b'_0})<_b \cdots <_b v_b(X_{b'_t})$.
\begin{definition}
The allocation $Y$ {\it dominates} $X$ with respect to $\psi$ if $\psi(Y)$ is lexicographic larger than $\psi(X)$.
\end{definition}
That is, we prioritize $N_a$ over $N_b$, compare agents in $N_a$ by the leximin ordering, and second compare agents in $N_b$ by the leximin ordering.
Note that if $Y$ Pareto dominates $X$, then $Y$ dominates $X$ with respect to $\psi$ but not vice versa.
Our goal is to show the following theorem.
\end{comment}
\begin{theorem}\label{thm: psi}
Let $X$ be a partial EFX allocation.
Then, there exists an EFX allocation $Y$ such that $Y \succ_{\rm p.lexmin} X$.\end{theorem}
If Theorem~\ref{thm: psi} holds, then since there are finitely many allocations, there must exist a complete EFX allocation, and thus Theorem~\ref{thm: two} holds.

We say that an allocation $X$ is {\it semi-EFX} if there can be EFX envy only among agents belonging to $N_b$ in $X$, i.e., no agent belonging to $N_a$ EFX envies any agents, and no agent belonging to $N_b$ EFX envies any agent belonging to $N_a$ in $X$.
The following lemma shows that if we have a semi-EFX allocation, then we can obtain an EFX allocation such that all the agents in $N_a$ and $b_0\in N_b$ do not change their bundles.
\begin{lemma}\label{lem: semi}
Let $X$ be a semi-EFX allocation such that $X_{\b_0} \lb X_{\b_1} \lb \cdots \lb X_{\b_t}$.
Then, there exists an EFX allocation $Y$ such that $Y_{a_i}=X_{a_i}$ for any agent $a_i\in N_a$, and $Y_{b_0}=X_{b_0}$.
\end{lemma}
\begin{proof}
If $X$ is EFX, then the lemma obviously holds.
Assume that $X$ is not EFX.
Then, for some two agents $b_i, b_j \in N_b$, $b_i$ EFX envies $b_j$ in $X$.
Thus, there exists an item $h \in X_{b_j}$ such that $X_{b_i} <_b X_{b_j} \setminus h$.
We define a new allocation $X'$ as $X'_{b_j}=X_{b_j} \setminus h$, and $X'_{k}=X_k$ for any $k\in N\setminus {b_j}$.
Then, $X'$ is also semi-EFX.
Indeed, since we only change $b_j$'s bundle, it suffice to consider EFX envy from or to $b_j$.
Since $b_j$'s bundle is a subset of $X_{b_j}$, and valuations are monotone, agents who do not EFX envy $b_j$ in $X$ do not EFX envy $b_j$ either in $X'$.
In addition, since $X_{b_i}=X'_{b_i} <_b X'_{b_j}$, and $b_i$ does not EFX envy any agent belonging to $N_a$ in $X$, $b_j$ does not EFX envy any agent in $N_a$ either in $X'$.
Therefore $X'$ is also semi-EFX.
Furthermore since $b_j$ is not $b_0$, we have $X'_{b_0}=X_{b_0}$.
If $X'$ is EFX, then we are done.
Otherwise, since the number of all items allocated is decreasing, we can continue this way to obtain an EFX allocation $Y$ such that $Y_{a_i}=X_{a_i}$ for any agent $a_i\in N_a$, and $Y_{b_0}=X_{b_0}$.
\end{proof}
We are now ready to prove Theorem~\ref{thm: psi}.
\begin{proof}[Proof of Theorem~\ref{thm: psi}]
Let $X$ be a partial EFX allocation and let $g$ be an unallocated item.
Define $N_\a = \{\a_0, \a_1, \ldots , \a_s\}$ and $N_\b = \{\b_0, \b_1, \ldots , \b_t\}$, where $X_{\a_0} \la X_{\a_1} \la \cdots \la X_{\a_s}$
and $X_{\b_0} \lb X_{\b_1} \lb \cdots \lb X_{\b_t}$.
If there exists a PI cycle in $M_X$, then we are done by Lemma~\ref{lem: PI}.
Assume that there is no PI cycle in $M_X$.
We first show that $a_0$ $g$-decomposes $b_0$.
By the assumption, neither $a_0$ nor $b_0$ is a self $g$-champion.
If $a_0$ envies $b_0$, then every agent other than $a_0$ is envied by some other agents.
Some agent $g$-champions $a_0$.
No matter who it is, there exists a PI cycle in $M_X$, and this is a contradiction.
Thus, $a_0$ does not envy $b_0$.
Now, some agent $i$ $g$-champions $b_0$.
If $i \in N_b$, then since $X_{b_0} \le_b X_{i}$, we have $\kappa_{X}(b_0, X_{b_0}\cup g)\le \kappa_{X}(i, X_{b_0}\cup g)=\kappa_{X}(X_{\b_0}\cup g)$.
This implies that $\b_0$ is a self $g$-champion, and this is a contradiction.
Hence, we have $i\in N_a$.
Then, since $X_{a_0} \le_a X_i$, we have $\kappa_{X}(a_0, X_{b_0}\cup g)\le \kappa_{X}(i, X_{b_0}\cup g)=\kappa_{X}(X_{\b_0}\cup g)$.
Hence, $a_0$ $g$-champions $b_0$.
As a result, $a_0$ $g$-decomposes $b_0$ by Lemma~\ref{lem: decompose}.
Therefore, $X_{b_0}$ is decomposed into top and bottom half-bundles.
Let $T_{\b_0}$ and $B_{\b_0}$ be the top and bottom half-bundles of $X_{b_0}$. 
Figure~\ref{fig: 1} partially illustrates $M_X$.
For $0\le i \le s$, we define $U_{a_i}\subseteq X_{a_i}$ as follows:
\begin{align*}
 U_{a_i} &= \left\{
    \begin{array}{ll}
      X_{a_i} & {\rm if}~ b_0 ~{\rm does~not~envy}~ a_i~{\rm in}~X,\\
     \hat{X}_{a_i} & {\rm otherwise},
    \end{array}
  \right. &\\
\end{align*}
where, $ \hat{X}_{a_i}$ is a maximum cardinality proper subset of $X_{a_i}$ maximizing $v_b( \hat{X}_{a_i})$.
Note that we have $|\hat{X}_{a_i}|=|X_{a_i}|-1$.
Consider $Z =\max_{\b}\{ T_{\b_0}\cup g, U_{\a_0},\dots,U_{\a_s}\}$.
We define a new allocation $X'$ as follows: 

\begin{figure}[tbp]
   \includegraphics[width=70mm]{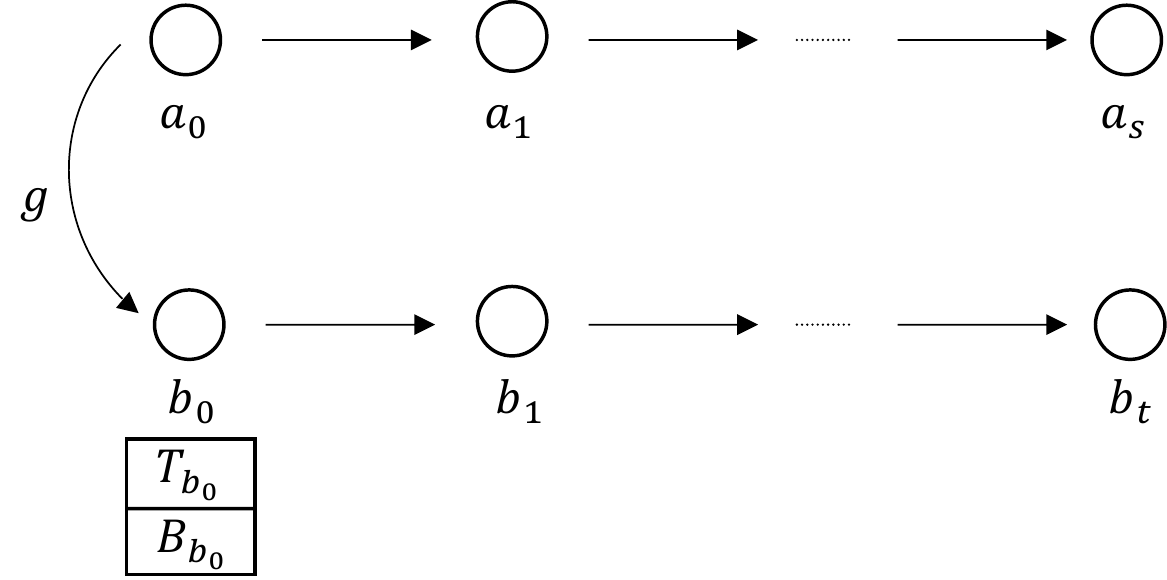}
   \caption{The champion graph $M_X$ (the edges are only partially drawn) in the proof of Theorem~\ref{thm: psi}. Every agent other than $a_0$ and $b_0$ is envied by some other agents, and $a_0$ $g$-decomposes $b_0$ in $M_X$.}

  \label{fig: 1}

\end{figure}

\begin{align*}
X'_{\a_i} &= \left\{
    \begin{array}{ll}
      T_{\b_0}\cup g & {\rm if}~Z=U_{\a_i}\\
      X_{\a_i} & {\rm otherwise}
    \end{array}
  \right. & ~{\rm for}~ 0\le i \le s ,\\
X'_{\b_0} &= Z, &\\
X'_{\b_j} &= X_{\b_j} & ~{\rm for}~ 1\le j \le t.&
\end{align*}
We can easily check that $X'$ is a legal allocation.
That is, $X'_{a_0},\dots,X'_{a_s},X'_{b_0},\dots,X'_{b_t}$ is a partition of a subset of $M$.
We show that $X'$ is semi-EFX.
\begin{itemize}
\item{\it Any two agents in $N_a$ do not EFX envy each other}:
Note that since $X_{a_0} <_a T_{b_0}\cup g$ and $X_{a_0} <_a X_{a_k}$ for $1\le k \le s$, we have $X_{a_0} \le_a X'_{a_i}$ for $0\le i \le s$.
Let $a_i$ and $a_{i'}$ be two agents in $N_a$.
If $X'_{a_{i'}}=X_{a_{i'}}$, then since $X_{a_0} \le_a X'_{a_i}$ and by the fact that $X$ is EFX, $a_i$ does not EFX envy $a_{i'}$ in $X'$. 
If $X'_{a_{i'}}=T_{b_0}\cup g$, then since $X_{a_0} \le_a X'_{a_i}$ and $a_0$ does not EFX envy $T_{b_0}\cup g$ in $X$, $a_i$ does not EFX envy $a_{i'}$ in $X'$. 
\item{\it Any agent in $N_a$ does not EFX envy any agent in $N_b\setminus b_0$}:
Since $X_{a_0} \le_a X'_{a_i}$ for $0\le i \le s$, $X$ is EFX, and any agent in $N_b\setminus b_0$ does not change her bundle, any agent in $N_a$ does not EFX envy any agent in $N_b\setminus b_0$.
\item{\it Any agent in $N_a$ does not EFX envy $b_0$}:
If $Z=T_{b_0}\cup g$, then since $X_{a_0} \le_a X'_{a_i}$ for $0\le i \le s$ and $a_0$ does not EFX envy $T_{b_0}\cup g$ in $X$, any agent in $N_a$ does not EFX envy $b_0$ in $X'$.
If $Z=U_{a_k}$ for some $0\le k \le s$, then since $X_{a_0} \le_a X'_{a_i}$ for $0\le i \le s$, and $a_0$ does not EFX envy $U_{a_k} \subseteq X_{a_k}$, 
any agent in $N_a$ does not EFX envy $b_0$ in $X'$.
\item{\it Any agent in $N_b\setminus b_0$ does not EFX envy any agent in $N_a$}:
Let $a_i$ be any agent in $N_a$ and let $b_j$ be any agent in $N_b\setminus b_0$. 
If $X'_{\a_i}=X_{\a_i}$, then since $X'_{b_j}=X_{b_j}$ and $X$ is EFX, $b_j$ does not EFX envy $a_i$.
If $X'_{\a_i}=T_{b_0}\cup g$, then since $b_0$ is not a self $g$-champion in $X$, we have $T_{\b_0}\cup g <_b X_{b_j}$.
Thus $b_j$ does not envy $a_i$ in $X'$.
\item{\it $b_0$ does not EFX envy any agent in $N_a$}:
Let $a_i$ be any agent in $N_a$.
If $X'_{a_i}= T_{b_0}\cup g$, then since $Z=\max_{\b}\{ T_{\b_0}\cup g, U_{\a_0},\dots,U_{\a_s}\} \ge_b T_{b_0}\cup g$, $b_0$ does not envy $a_i$ in $X'$.
If $X'_{a_i}=X_{a_i}$, then since $Z=\max_{\b}\{ T_{\b_0}\cup g, U_{\a_0},\dots,U_{\a_s}\} \ge_b U_{a_i}$, for any proper subset $S$ of $X_{a_i}$, we have $Z \ge_b U_{\a_i} \ge_b S$ by the definition of $U_{\a_i}$ and $\hat{X}_{a_i}$.
Thus $b_0$ does not EFX envy $a_i$ in $X'$.
\end{itemize}
Therefore $X'$ is semi-EFX.
By Lemma~\ref{lem: semi}, there exists an EFX allocation $X''$ such that $X''_{a_i}=X'_{a_i}$ for $0\le i \le s$, and $X''_{b_0}=X'_{b_0}$.
We discuss in the following three cases.
\begin{description}
\item[Case 1:]  $Z=U_{\a_0}$

In this case, we have $X''_{\a_0}=X'_{\a_0}=T_{\b_0}\cup g >_{\a} X_{\a_0}$ and $X''_{\a_k}=X'_{\a_k}=X_{\a_k}$ for $1\le k \le s$.
Thus, we have $X'' \succ_{\rm p.lexmin} X$, and we are done.
\item[Case 2:]  $Z=T_{\b_0}\cup g$

In this case, since we have $X''_{\a_0}=X'_{\a_0}=X_{\a_0} <_{\a} T_{\b_0}\cup g=X'_{\b_0}=X''_{b_0}$, 
$\a_0$ envies $\b_0$ in $X''$.
Thus, every agent other than $a_0$ is envied by some other agents in $X''$.
By the fact that $B_{b_0}\neq \emptyset$, there is an unallocated item $g'\in B_{b_0}$.
Then, some agent $l$ $g'$-champions $a_0$ (see Figure~\ref{fig: 2}).
If $l=a_i \in N_a$, then by following agents in $N_a$ backwards we obtain a PI cycle $a_0 \rightarrow\cdots \rightarrow a_{i-1}\rightarrow a_i\xrightarrow{g'} a_0$ in $M_{X''}$.
If $l\in N_b$, then by following agents in $N_b$ backwards we also obatain a PI cycle $a_0 \rightarrow b_0 \rightarrow \cdots \rightarrow  l\xrightarrow{g'} a_0$ in $M_{X''}$.
Therefore, in either case, 
there exists a PI cycle containing $a_0$ in $M_{X''}$.
By Lemma~\ref{lem: PI}, there exists an EFX allocation $X'''$ such that $X'''_{\a_0}>_a X''_{\a_0}=X'_{\a_0}=X_{\a_0}$ and $X'''_{\a_k}\ge_a X''_{\a_k}=X'_{\a_k}=X_{\a_k}$ for $1\le k \le s$.
Therefore, we have $X''' \succ_{\rm p.lexmin} X$, and we are done.
\begin{figure}[t]
 \begin{tabular}{cc}
 \begin{minipage}[t]{0.5\hsize}
  \begin{center}
   \includegraphics[width=60mm]{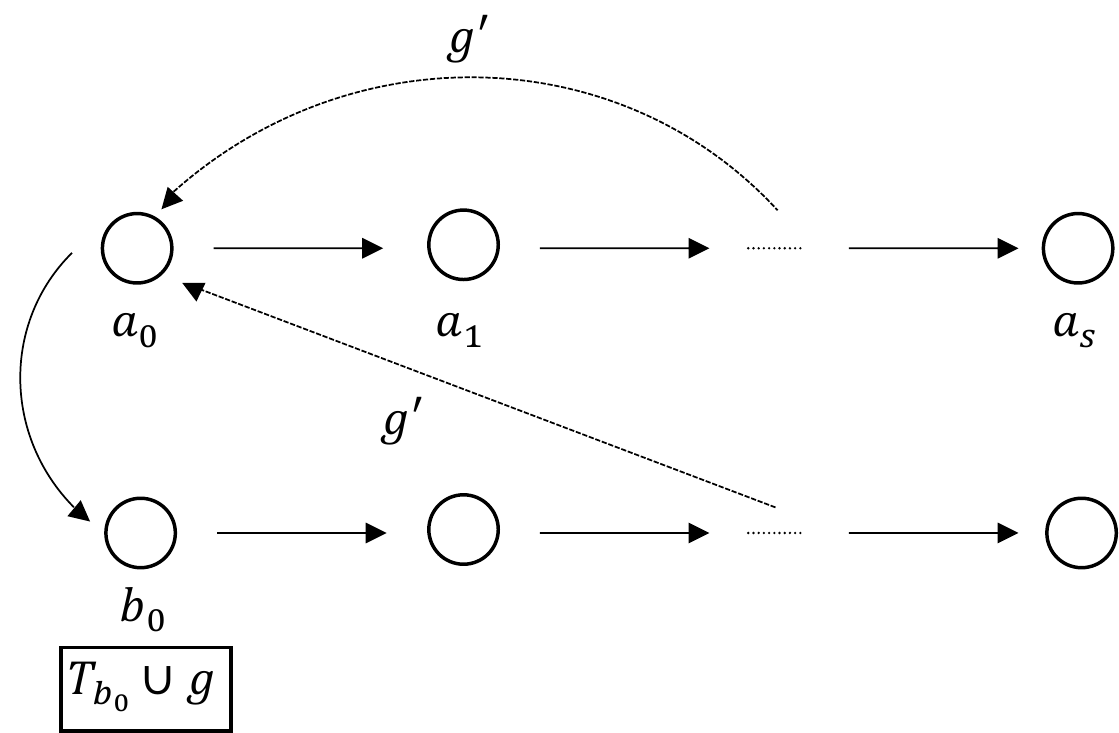}
   \captionsetup{width=.95\linewidth}
   \caption{The champion graph $M_{X''}$ (the edges are only partially drawn) in Case 2. Every agent other than $a_0$ is envied by some other agents, and some agent $g'$-champions $a_0$ in $M_{X''}$.}
  \label{fig: 2}
   \end{center}
 \end{minipage}
 \hfill
 \begin{minipage}[t]{0.5\hsize}
   \begin{center}
   \includegraphics[width=60mm]{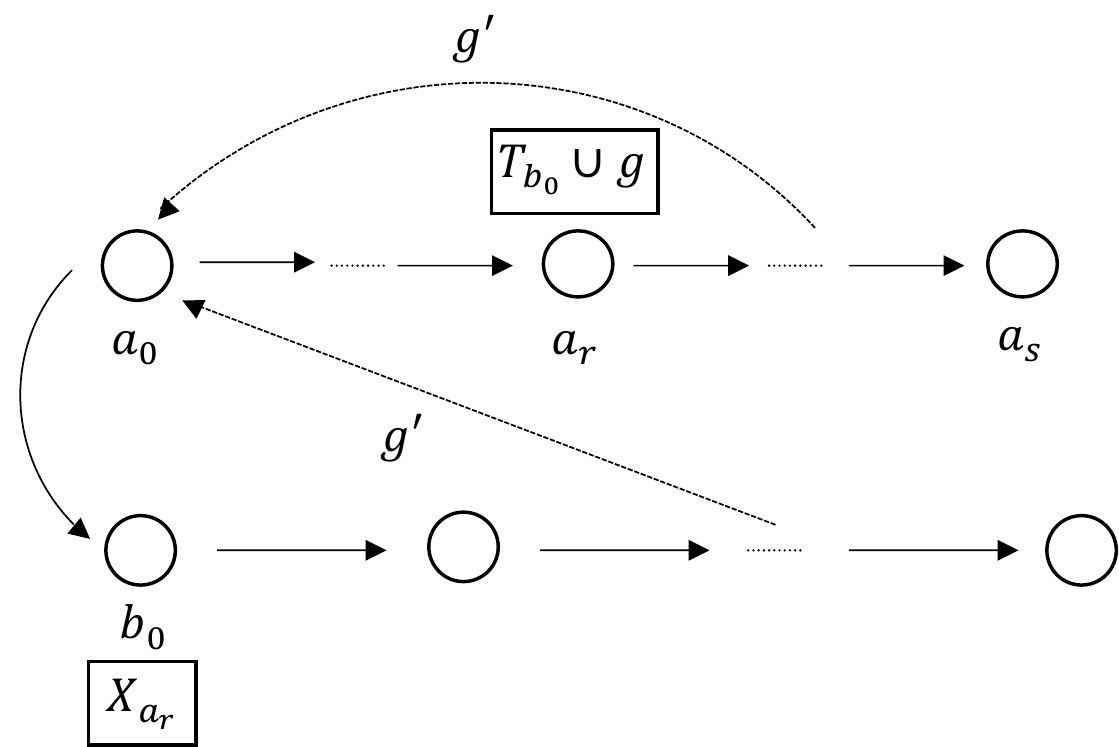}
   \captionsetup{width=.95\linewidth}
   \caption{The champion graph $M_{X''}$ (the edges are only partially drawn) in the case where $U_{a_r}=X_{a_r}$ in Case 3 . Every agent other than $a_0$ is envied by some other agents, and some agent $g'$-champions $a_0$ in $M_{X''}$.}
     \label{fig: 3}
    \end{center}
 \end{minipage}\\

   \begin{minipage}[t]{0.5\hsize}
  \begin{center}
   \includegraphics[width=55mm]{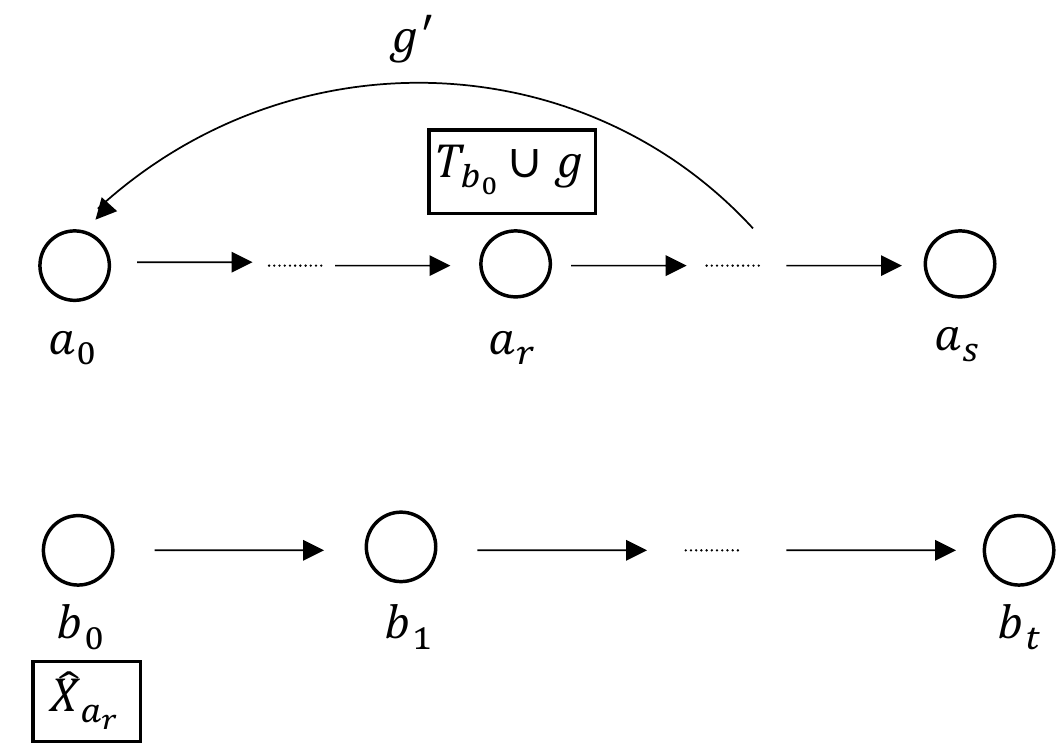}
   \captionsetup{width=.95\linewidth}
   \caption{The champion graph $M_{X'}$ (the edges are only partially drawn) in the case where $U_{a_r}=\hat{X}_{a_r}$, and some agent in $N_a$ $g'$-champions $a_0$ in Case 3.}
  \label{fig: 4}
   \end{center}
 \end{minipage}
  \hfill
  \begin{minipage}[t]{0.5\hsize}
  \begin{center}
   \includegraphics[width=60mm]{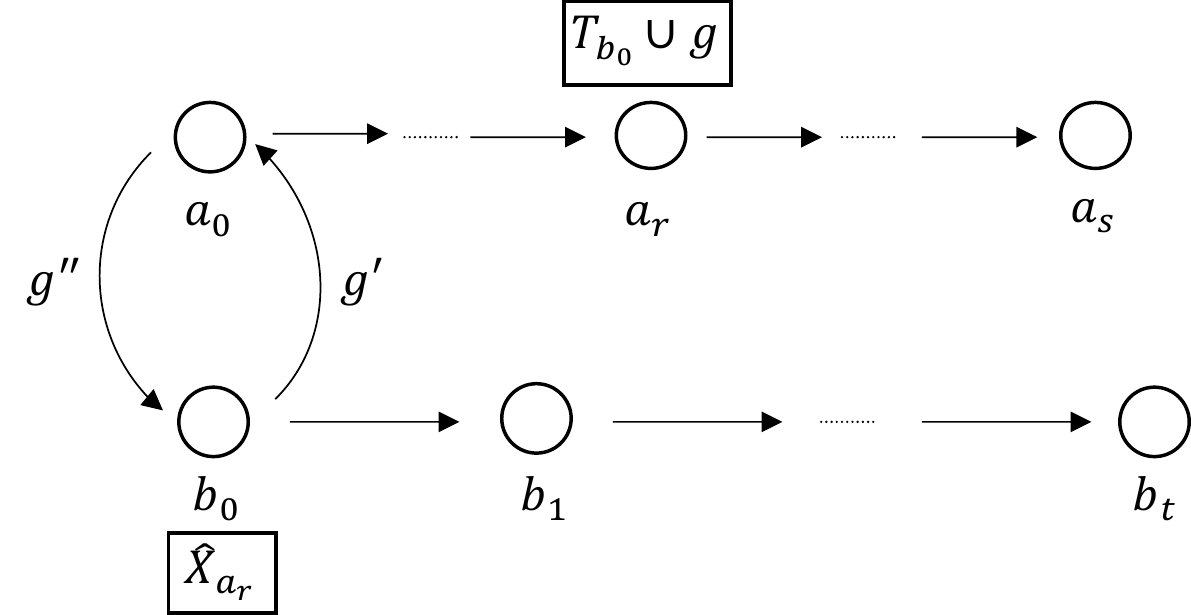}
   \captionsetup{width=.95\linewidth}
   \caption{The champion graph $M_{X'}$ (the edges are only partially drawn) in the case where $U_{a_r}=\hat{X}_{a_r}$, and some agent in $N_b$ $g'$-champions $a_0$ in Case 3. $b_0$ $g'$-champions $a_0$, and $a_0$ $g''$-champions $b_0$ in $M_{X'}$.}
  \label{fig: 5}
   \end{center}
 \end{minipage}
 \end{tabular}
\end{figure}
\item[Case 3:] $Z=U_{\a_r}$ for some $1\le r \le s$

In this case, if $U_{a_r}=X_{a_r}$, then since we have $X''_{a_0}=X'_{a_0}=X_{a_0} <_a X_{a_r} = X'_{b_0}=X''_{b_0}$, $a_0$ envies $b_0$ in $X''$.
By the fact that $B_{b_0}\neq \emptyset$, there is an unallocated item $g'\in B_{b_0}$.
In a similar way to Case $2$, the fact that some agent $g'$-champions $a_0$ implies that there exists a PI cycle containing $a_0$ in $M_{X''}$ (see Figure~\ref{fig: 3}).
By Lemma~\ref{lem: PI}, there exists an EFX allocation $X'''$ such that $X'''_{\a_0}>_a X''_{\a_0}$ and $X'''_{\a_k}\ge_a X''_{\a_k}$ for $1\le k \le s$.
Thus, we have $X'''_{\a_k}\ge_a X''_{\a_k} = X'_{a_k} = X_{a_k} >_a X_{a_0}$ for $1\le k \le s$ with $k\neq r$, and $X'''_{a_r} \ge_a X''_{a_r} =X'_{a_r} = T_{b_0} \cup g >_a X_{a_0}$.
That is, $X'''_{a_i} >_a X_{a_0}$ for $0\le i \le s$.
Therefore, we have $X''' \succ_{\rm p.lexmin} X$, and we are done.

If $U_{a_r}=  \hat{X}_{a_r}$, then we consider semi-EFX allocation $X'$, not $X''$ in this case.
some agent $l$ $g'$-champions $a_0$ in $M_{X'}$.
If $l \in N_a$, then by following agents in $N_a$ backwards we obtain a PI cycle $a_0 \rightarrow\cdots \rightarrow  l\xrightarrow{g'} a_0$ in $M_{X''}$ (see Figure~\ref{fig: 4}).
By Lemma~\ref{lem: PI}, there exists a semi-EFX allocation $X'''$ such that $X'''_{\a_0}>_a X'_{\a_0}$ and $X'''_{\a_k} \ge_a X'_{a_k}$ for $1\le k\le s$.
By a similar argument as above, we have $X'''_{a_i} >_a X_{a_0}$ for $0\le i \le s$.
By Lemma~\ref{lem: semi}, there exists an EFX allocation $X''''$ such that $X''''_{a_i}=X'''_{a_i}$ for $0\le i \le s$.
Therefore, we have $X''''_{a_i} >_a X_{a_0}$ for $0\le i \le s$, thus we have $X'''' \succ_{\rm p.lexmin} X$, and we are done.

If $l =b_j \in N_b$, 
then since $X'_{b_0} \le_b X'_{b_j}$, we have $\kappa_{X'}(b_0, X'_{a_0}\cup g') \le \kappa_{X'}(b_j, X'_{a_0}\cup g')=\kappa_{X'}(X'_{a_0}\cup g')$.
Thus $b_0$ $g'$-champions $a_0$ in $X'$.
Since $U_{a_r}=\hat{X}_{a_r}$, there exists an unallocated item $g''\in X_{a_r} \setminus \hat{X}_{a_r}$.
Note that we have $X_{a_r}=X'_{b_0}\cup g''$.
We claim that $a_0$ $g''$-champions $b_0$ in $X'$.
Indeed, since any agent $u\in N \setminus \{a_r, b_0\}$ does not change her bundle, and $X$ is EFX, $u$ does not EFX envy $X_{a_r}=X'_{b_0}\cup g''$ in $X'$.
In addition, since $X'_{a_r} =T_{b_0} \cup g >_a X_{a_0}$, and $a_0$ does not EFX envy $X_{a_r}$ in $X$, $a_r$ does not EFX envy $X_{a_r}$ in $X'$.
Furthermore, since $X'_{b_0}=U_{a_r}=  \hat{X}_{a_r}$ is a maximum cardinality proper subset of $X_{a_r}$ maximizing $v_b( \hat{X}_{a_r})$, $b_0$ does not EFX envy $X_{a_r}$ in $X'$.
To sum up, for any proper subset $S$ of $X'_{b_0}\cup g''$, any agent in $N$ does not envy $S$ in $X'$.
Furthermore, since we have $X'_{a_0} =X_{a_0} <_a X_{a_r}=X'_{b_0} \cup g''$, 
$a_0$ envies $X'_{b_0} \cup g''$ in $X'$.
Thus, since $\kappa_{X'}(a_0, X'_{b_0}\cup g'') = |X'_{b_0} \cup g''| \le \kappa_{X'}(w, X'_{b_0}\cup g'')$ for any $w\in N$, $a_0$ is a most envious agent for $X'_{b_0}\cup g''$.
That is, $a_0$ $g''$-champions $b_0$ in $X'$ (see Figure~\ref{fig: 5}).
We now have a PI cycle $a_0 \xrightarrow{g''} b_0 \xrightarrow{g'} a_0$ in $M_{X'}$, and by Lemma~\ref{lem: PI}, we obtain a semi-EFX allocation $X'''$ such that $X'''_{a_0} >_a X'_{a_0}$ and $X'''_{a_k} = X'_{a_k}$ for $1\le k \le s$.
By a similar argument as above, we have $X'''_{a_i} >_a X_{a_0}$ for $0\le i \le s$.
By Lemma~\ref{lem: semi}, there exists an EFX allocation $X''''$ such that $X''''_{a_i}=X'''_{a_i}$ for $0\le i \le s$.
Therefore, we have $X''''_{a_i} >_a X_{a_0}$ for $0\le i \le s$, thus we have $X'''' \succ_{\rm p.lexmin} X$, and we are done.
\end{description}
\end{proof}

\section{Existence of EFX with at most $n+3$ items}
\label{sec: m=n+3}
In this section, we investigate the setting with at most $n+3$ items, and prove Theorem~\ref{thm: m=n+3}.
We assume that $m\le n+3$ in this section.
We use the lexicographic potential function as in Section~$\ref{sec: n-2}$.
By Lemma~\ref{lem: progress}, in order to prove Theorem~\ref{thm: m=n+3}, it suffices to show that for every partial EFX allocation $X$ with at least $1$ unallocated item, there exists an EFX allocation $Y$ such that $Y\succ_{\rm lex} X$.

Given an allocation, 
let an {\it 1-bundle} be a bundle such that the cardinality is exactly one.
Let an {\it 1-agent} be an agent who has an 1-bundle.
Note that since valuation is normalized, no agent EFX envies 1-agents.
The following lemma shows that in a partial allocation $X$, if $M_X$ has no PI cycles, then there are exactly two non-1-agents.
Moreover, such non-1-agents have bundles of cardinality two, which implies that $m=n+3$ and there is one unallocated item in $X$.
\begin{lemma}\label{lem: 2}
Let $X$ be a partial allocation.
If $M_X$ has no PI cycles, then there are exactly two non-1-agents.
Furthermore, every 1-agent is envied by some other agents, and every non-1-agent is not envied by any agents.
\end{lemma}
\begin{proof}
Since there is no self-champion in $M_X$, every agent has at least one item.
Thus, since there is at least one unallocated item, the fact that $m\le n+3$ implies that the number of non-1-agents is at most two.
We claim that for any 1-agent $i$, some agent envies $i$.
Indeed, for any unallocated item $g$, some agent $j$ $g$-champions $i$.
Let $P$ be a minimum preferred set of $j$ for $X_i\cup g$.
Since both $i$ and $j$ are not self $g$-champions, we obtain $P\neq \{g\}$ and $X_i\cup g$.
Hence, $P=X_i$, and thus $j$ envies $i$.
If there are no non-1-agents, 
then every agent is envied by some other agents.
In other words, every agent has an in-coming envy edge, and thus we have an envy-cycle in $M_X$, and this contradicts the assumption.
If there is exactly one non-1-agent, then
for any unallocated item $g$, some agent $g$-champions the non-1-agent.
No matter who it is, 
we have a PI cycle in $M_X$, and this contradicts the assumption as well.
Therefore, we conclude that the number of non-1-agents is exactly two, and every 1-agent is envied by some other agents.
For the last statement of the lemma, if both of non-1-agents are envied by some other agents, then every agent has an in-coming envy edge, and thus we have an envy-cycle in $M_X$.
This contradicts the assumption.
If exactly one of non-1-agents are envied by some other agents, then every agent other than another non-1-agent has an in-coming envy edge.
For any unallocated item $g$, some agent $g$-champions the non-1-agent.
No matter who it is, 
we have a PI cycle in $M_X$, and this contradicts the assumption.
\end{proof}
For a dicycle $C=a_1\rightarrow a_2\rightarrow \cdots \rightarrow a_k \rightarrow a_1$, let $\mathbf{succ}(a_i)$ and $\mathbf{pred}(a_i)$ denote the successor and predecessor of $a_i$ along the cycle, respectively.
Let $V[C]$ be all vertices on $C$.
For vertices $i, j$ on $C$, Let $[i,j]$ denote a set of all vertices on $P^C_{i,j}$, where $P^C_{i,j}$ is the unique $(i,j)$-dipath on $C$.
We are now ready to prove Theorem~\ref{thm: m=n+3}.
We fix an arbitrary ordering of the agents.
\begin{proof}[Proof of Theorem~\ref{thm: m=n+3}]
Let $X$ be a partial allocation.
It suffices to show that there exists an EFX allocation $Y$ such that $Y \succ_{\rm lex} X$ by Lemma~\ref{lem: progress}.
If $M_X$ has a PI cycle, then there exists an EFX allocation $Y$ that Pareto dominates $X$ by Lemma~\ref{lem: PI}, and we are done.
Assume that $M_X$ has no PI cycles.
Then, by Lemma~\ref{lem: 2}, the number of non-1-agents is exactly two.
In this situation, we have one unallocated item $g$, and non-1-agents have bundles of cardinality two.
Since every 1-agent has in-coming envy edges in $M_X$ by Lemma~\ref{lem: 2} and every non-1-agent has in-coming $g$-champion edges in $M_X$, 
these edges induce a dicycle $C$ such that the number of champion edges in $C$ is at most two.
If the dicycle $C$ has at most one champion edge, 
then by Corollary~\ref{cor: PI}, we have a PI cycle, and we are done.
Suppose that the number of champion edges in $C$ is exactly two.
In this situation, $C$ has exactly two $g$-champion edges $i_1\xrightarrow{g} i_2$ and $i_3\xrightarrow{g} i_4$ (possibly, $i_1=i_4$ and/or $i_2=i_3$), and $i_2$ and $i_4$ are non-1-agents.
Let $i\in \{i_1, i_2, i_3, i_4\}$ be the last agent with the fixed ordering.
That is, $i$ is the least important agent among agents who are endpoints of champion edges in $C$.
We discuss in two cases.
\begin{figure}[t]
 \begin{tabular}{cc}
 \begin{minipage}[t]{0.5\hsize}
  \begin{center}
   \includegraphics[width=70mm]{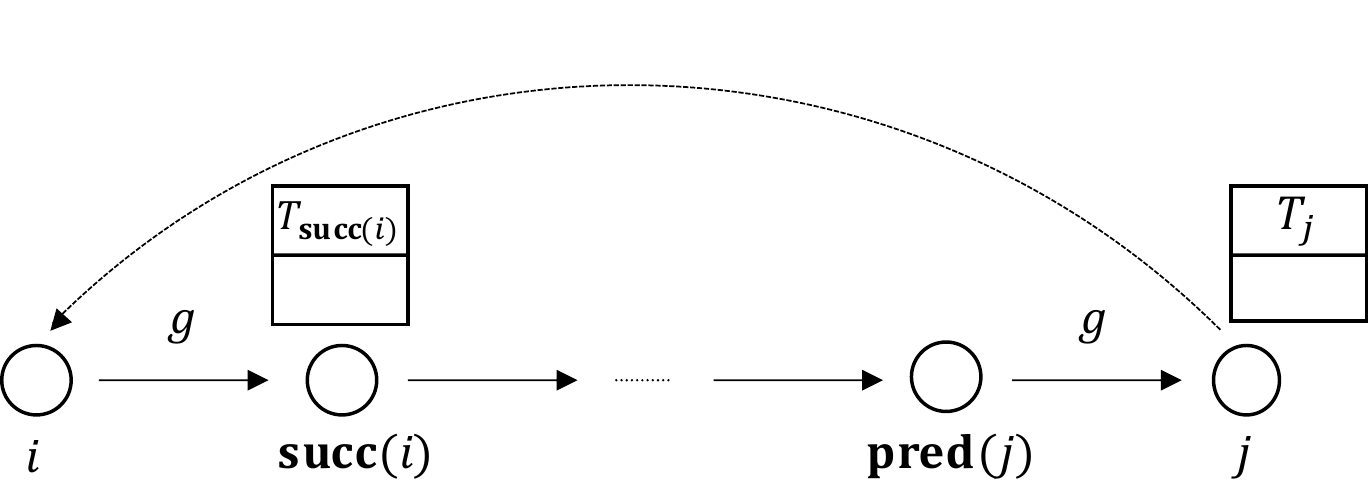}
   \captionsetup{width=.95\linewidth}
   \caption{The dicycle $C$ in Case 1. $i$ $g$-decomposes $\mathbf{succ}(i)$, and $\mathbf{pred}(j)$ $g$-decomposes $j$. The dotted arrow represents $(j, i)$-dipath on $C$.}
  \label{fig: 6}
   \end{center}
 \end{minipage}
 \hfill
 \begin{minipage}[t]{0.5\hsize}
   \begin{center}
   \includegraphics[width=70mm]{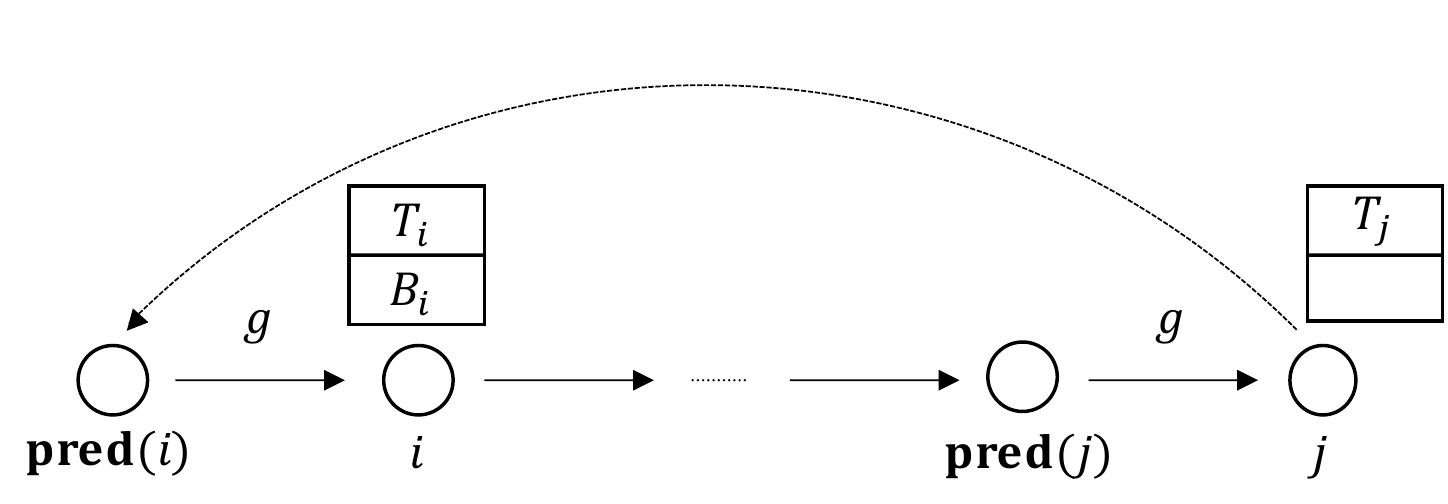}
   \captionsetup{width=.95\linewidth}
   \caption{The dicycle $C$ in Case 2. $\mathbf{pred}(i)$ $g$-decomposes $i$, and $\mathbf{pred}(j)$ $g$-decomposes $j$. The dotted arrow represents $(j, \mathbf{pred}(i))$-dipath on $C$.}
     \label{fig: 7}
    \end{center}
 \end{minipage}

 \end{tabular}
\end{figure}

\begin{description}
\item[Case 1:] $i$ has an in-coming envy edge in $C$

Since $i$ is envied by some other agents, $i$ is an $1$-agent by Lemma~\ref{lem: 2}.
Since $i$ is an endpoint of champion edges in $C$, $i$ $g$-champions $\mathbf{succ}(i)$, and $\mathbf{succ}(i)$ is a non-1-agent.
Let $j$ be a non-1-agent who is not $\mathbf{succ}(i)$.
Note that $j$ is before $i$ in the fixed ordering.
Since there are no envy edges $i\rightarrow \mathbf{succ}(i)$ and $\mathbf{pred}(j) \rightarrow j$ (possibly, $\mathbf{succ}(i)=\mathbf{pred}(j)$), and there are no self $g$-champions in $M_X$, 
by Lemma~\ref{lem: decompose}, $i$ $g$-decomposes $\mathbf{succ}(i)$, and $\mathbf{pred}(j)$ $g$-decomposes $j$.
Let $T_{\mathbf{succ}(i)}$ and $T_{j}$ be the top half-bundles of $X_{\mathbf{succ}(i)}$ and $X_j$, respectively (see Figure~\ref{fig: 6}).
Consider $Z=\max_{i} \{X_{\mathbf{succ}(i)}, T_j\cup g\}$.
\begin{description}
\item[Case 1-1:] $Z=T_j\cup g$

We define a new allocation $X'$ as follows: 
\begin{align*}
X'_k &= \left\{ 
\begin{array}{ll}
Z & {\rm if}~k=i, \\
X_{\mathbf{succ}(k)} & {\rm if}~k\in [j, \mathbf{pred}(i)],\\
X_{k} & {\rm otherwise}.
\end{array}
 \right. &
\end{align*}
We show that $X'$ is EFX.
Note that only $i$ and $\mathbf{succ}(i)$ are non-1-agents in $X'$.
Since no agent EFX envies 1-agents, it suffices to check that no agent EFX envies $X'_i=T_j \cup g$ and $X'_{\mathbf{succ}(i)}=X_{\mathbf{succ}(i)}$.
Since any agent other than $i$ is not worse off than in $X$, no agent EFX envies $T_j\cup g$ in $X$, and $X$ is EFX, 
any agent other than $i$ does not EFX envy $T_j \cup g$ and $X_{\mathbf{succ}(i)}$ in $X'$.
Furthermore by the definition of $Z$, $i$ does not envy $X_{\mathbf{succ}(i)}$ in $X'$.
Therefore $X'$ is EFX.
Furthermore, since any agent other than $i$ is not worse off, and $j$ before $i$ in the fixed ordering is strictly better off, we conclude that $X' \succ_{\rm lex} X$, and we are done.

\item[Case 1-2:] $Z=X_{\mathbf{succ}(i)}$

We define a new allocation $X'$ as follows: 
\begin{align*}
X'_k &= \left\{ 
\begin{array}{ll}
T_j\cup g & {\rm if}~k=\mathbf{pred}(j), \\
X_{\mathbf{succ}(k)} & {\rm if}~k\in V[C] \setminus \mathbf{pred}(j),\\
X_{k} & {\rm otherwise}.
\end{array}
 \right. &
\end{align*}

We show that $X'$ is EFX.
It suffices to check that no agent EFX envies $X'_{\mathbf{pred}(j)}=T_j \cup g$ and $X'_i=X_{\mathbf{succ}(i)}$.
Since any agent other than $i$ is not worse off than in $X$, no agent EFX envies $T_j\cup g$ in $X$, and $X$ is EFX, 
any agent other than $i$ does not EFX envy $T_j \cup g$ and $X_{\mathbf{succ}(i)}$ in $X'$.
Furthermore by the definition of $Z$, $i$ does not envy $T_j\cup g$ in $X'$.
Therefore $X'$ is EFX.
Furthermore, since any agent other than $i$ is not worse off than in $X$, and $j$ before $i$ in the fixed ordering is strictly better off, we conclude that $X' \succ_{\rm lex} X$, and we are done.
\end{description}

\item[Case 2:] $i$ has an in-coming $g$-champion edge in $C$.

In this case, $i$ is a non-1-agent.
Let $j$ be a non-1-agent who is not $i$.
Since there are no envy edges $\mathbf{pred}(i) \rightarrow i$ and $\mathbf{pred}(j)\rightarrow j$ (possibly, $i=\mathbf{pred}(j)$ and/or $j=\mathbf{pred}(i)$), and there are no self $g$-champions in $M_X$, by Lemma~\ref{lem: decompose}, $\mathbf{pred}(i)$ $g$-decomposes $i$, and $\mathbf{pred}(j)$ $g$-decomposes $j$.
Let $T_{i}$ and $T_{j}$ be the top half-bundles of $X_i$ and $X_j$, respectively, and let $B_i$ be the bottom half-bundles of $X_i$ (see Figure~\ref{fig: 7}).
Consider $Z=\max_{i} \{T_i \cup g, X_j\}$.

\begin{description}
\item[Case 2-1:] $Z=T_i \cup g$

We define a new allocation $X'$ as follows: 
\begin{align*}
X'_k &= \left\{ 
\begin{array}{ll}
Z & {\rm if}~k=i, \\
X_{k} & {\rm if}~k\neq i.
\end{array}
 \right. &
\end{align*}
We show that $X'$ is EFX.
It suffices to check that no agent EFX envies $X'_i=T_i \cup g$ and $X'_j=X_j$.
Since any agent other than $i$ does not change her bundle, no agent EFX envies $T_j\cup g$ in $X$, and $X$ is EFX, any agent other than $i$ does not EFX envy $T_i \cup g$ and $X_j$ in $X'$.
Furthermore by the definition of $Z$, $i$ does not envy $X_j$ in $X'$.
Therefore $X'$ is EFX.
In addition, since $T_i \cup g$ is a minimum preferred set of $\mathbf{pred}(i)$, $\mathbf{pred}(i)$ envies $i$ in $X'$.
Thus, any agent other than $j$ is envied by some other agents in $X'$.
We now have an unallocated item $g'\in B_{i}$.
Some agent $u$ $g'$-champions $j$.
If $u\in V[C]$, then by following agents on $V[C]$ backwards we have a PI cycle $j \rightarrow \cdots \rightarrow \mathbf{pred}(u) \rightarrow u \xrightarrow{g'} j$ in $M_{X'}$.
If $u\not\in V[C]$, then since there is no envy-cycle in $M_X$, and any agent other than $i$ does not change her bundle, there is no envy-cycle that does not contain $i$ in $M_{X'}$.
Thus, since any agent other than $j$ has an in-coming envy edge in $M_{X'}$, there exists agent $u'\in V[C]$ such that $u$ is reachable from $u'$ by following envy edges in $M_{X'}$. 
Then, by following agents in $V[C]$ from $u$ to $j$, we obtain a PI cycle $j\rightarrow \cdots \rightarrow u' \rightarrow \cdots \rightarrow u \xrightarrow{g'} j$ in $M_{X'}$
Therefore in any case, we have a PI cycle containing $j$ in $M_{X'}$, and thus there exists an EFX allocation $X''$ Pareto dominating $X'$ by Lemma~\ref{lem: PI}.
Agent $j$ is strictly better off than in $X'$, i.e., $X''_j >_j X'_j = X_j$.
Since any agent other than $i$ is not worse off than in $X$, 
and $j$ before $i$ in the fixed ordering is strictly better off, we conclude that $X'' \succ_{\rm lex} X$, and we are done.
\item[Case 2-2:] $Z=X_j$

We define a new allocation $X'$ as follows: 
\begin{align*}
X'_k &= \left\{ 
\begin{array}{ll}
Z & {\rm if}~k=i, \\
T_i\cup g & {\rm if}~k=\mathbf{pred}(i), \\
X_{\mathbf{succ}(k)} & {\rm if}~k\in [j, \mathbf{pred}(i)] \setminus \mathbf{pred}(i),\\
X_{k} & {\rm otherwise}.
\end{array}
 \right. &
\end{align*}
We show that $X'$ is EFX.
It suffices to check that no agent EFX envy $X'_{\mathbf{pred}(i)}=T_i \cup g$ and $X'_i=X_j$.
Since any agent other than $i$ does not change her bundle, no agent EFX envies $T_j\cup g$ in $X$, and $X$ is EFX, any agent other than $i$ does not EFX envy $T_i \cup g$ and $X_j$ in $X'$.
Furthermore by the definition of $Z$, $i$ does not envy $T_i \cup g$ in $X'$.
Therefore $X'$ is EFX.
Furthermore, since any agent other than $i$ is not worse off than in $X$, and $j$ before $i$ in the fixed ordering is strictly better off, we conclude that $X' \succ_{\rm lex} X$, and we are done.
\end{description}
\end{description}
\end{proof}

\section{Limitations of the Lexicographic Potential Function}
\label{sec: limit}
In Section~\ref{sec: n-2} and~\ref{sec: m=n+3}, 
in order to prove the existence of EFX, we show that given a partial EFX allocation $X$, there exists an EFX allocation $Y$ such that $Y \succ_{\rm lex} X$.
Recently, Chaudhury et al.~\cite{chaudhury2021improving} have shown that there does not always exist a lexicographically larger EFX allocation when $n=4$ for additive valuations.
In this section, we show that there does not always exist a lexicographically larger EFX allocation when $n=3$ for general valuations.
Thus, the approach using the lexicographic potential funtion is not sufficient to show the existence of EFX even when $n=3$ for general valuations.

The following theorem shows that there exist an instance and a partial EFX allocation $X$ such that no complete EFX allocation $Y$ such that $Y\succ_{\rm lex} X$.
\begin{theorem}
There exist an instance $I$ with three agents, $\{1,2,3\}$ with general valuations, seven items $\{g_i\mid i \in [7]\}$, and a partial EFX allocation $X$, such that in every complete EFX allocation, the valuation of agent $1$ will be strictly worse off than in $X$.
\end{theorem}
\begin{proof}
We partially define the conditions of each agent's valuation function.
Assume that agent $1$ has an additive valuation satisfying the following conditions: 
$v_1(g_1)=v_1(g_2) > 0, v_1(g_i)=0$ for $3\le i \le 7$.
Agent $2$ has a general valuation satisfying following four conditions: 
\begin{enumerate}
\item[(1)] $\{g_i\} <_2 \{g_1\}$ for $2\le i \le 7$
\item[(2)] $\{g_i, g_j\} <_2 \{g_1\}$ for $2\le i < j \le 7$ and $(i,j)\neq (3,4), (5,7)$
\item[(3)] $\{g_4, g_5, g_6\} <_2 \{g_1\} <_2 \{g_3,g_4\} <_2 \{g_5, g_7\}$
\item[(4)] $\{g_5, g_7\} <_2 \{g_1, g_i\}$ for $2\le i  \le 7$
\end{enumerate}

Similarly, agent $3$ has a general valuation satisfying the following four conditions: 
\begin{enumerate}
\item[(1')] $\{g_i\} <_3 \{g_1\}$ for $2\le i \le 7$
\item[(2')] $\{g_i, g_j\} <_3 \{g_1\}$ for $2\le i < j \le 7$ and $(i,j)\neq (5,6), (3,7)$
\item[(3')] $\{g_3, g_4, g_6\} <_3 \{g_1\} <_3 \{g_5, g_6\} <_3 \{g_3, g_7\}$
\item[(4')] $\{g_3, g_7\} <_3 \{g_1, g_i\}$ for $2\le i  \le 7$
\end{enumerate}
Note that all conditions do not violate the monotonicity of valuation functions.
We now consider a partial allocation $X=(\{g_1,g_2\}, \{g_3,g_4\}, \{g_5,g_6\})$.
We can easily check that $X$ is an EFX allocation.
Consider any complete EFX allocation $Y$.
We show that $Y_1 <_1 X_1$.
Assume that $X_1 \le_1 Y_1$.
Then, it must be $\{g_1, g_2\} \subseteq Y_1$ by the definition of $1$'s valuation.
If $\{g_1, g_2\} \subsetneq Y_1$, then at least one of agents $2$ and $3$ has a bundle of size at most $2$.
If $|Y_2|\le 2$, then since $Y_2 <_2 \{g_1, g_2\}$ by (2), (3), and (4), agent $2$ EFX envies agent $1$.
This is a contradiction.
The similar argument holds when $|Y_3|\le 2$. 
Thus, we have $Y_1=\{g_1, g_2\}$.
Therefore, by (1) and (1'), both agent $2$ and agent $3$ have bundles of size at least $2$.
This implies that $|Y_2|=2$ and $|Y_3|=3$, or $|Y_2|=3$ and $|Y_3|=2$.

If $|Y_2|=2$ and $|Y_3|=3$, then
since agent $2$ does not EFX envy agent $1$, $Y_2=\{g_3, g_4\}$ or $\{g_5, g_7\}$ by (2).
If $Y_2=\{g_3, g_4\}$ then $Y_3=\{g_5, g_6, g_7\}$, and agent $2$ EFX envies agent $3$ by (3).
This is a contradiction.
If $Y_2=\{g_5, g_7\}$ then $Y_3=\{g_3, g_4, g_6\}$, and agent $3$ EFX envies agent $1$ by (3').
This is a contradiction.
The similar argument holds when $|Y_2|=3$ and $|Y_3|=2$. 
As a result, we conclude that $Y_1 <_1 X_1$, and thus the value of agent $1$ will be strictly worse off than in $X$.
\end{proof}

\subsection*{Acknowledgments}
The author would like to thank Yusuke Kobayashi for his generous support and useful discussion.
This work was partially supported by the joint project of Kyoto University and Toyota Motor Corporation, titled ``Advanced Mathematical Science for Mobility Society''.
\bibliography{efx}
\bibliographystyle{plainurl}

\appendix %%%%%%%%%%%%%%%%%appendix%%%%%%%%%%%%%%%%%%%%

\section{Non-Degenerate instances}\label{sec: non-degenerate}
In this section, we show that in order to prove the existence of an EFX allocation, we may assume w.l.o.g.~that instances are non-degenerate.
Recall that an instance $I$ is a triple $\langle N, M, \mathcal{V} \rangle $, where $N$ is a set of agents, $M$ is a set of items and $\mathcal{V}=\{v_1,\dots,v_n\}$ is a set of valuation functions.
Let $M=\{g_0,\dots,g_{m-1}\}$ and let $\epsilon > 0$ be a positive real number.
We perturb an instance $I$ to $I_\epsilon =\langle N, M, \mathcal{V}_\epsilon \rangle $, where for any $v_i \in \mathcal{V}$ and any $S\subseteq M$, we define $v'_i \in \mathcal{V}_\epsilon$ by
$$v'_i(S)=v_i(S)+\epsilon\sum_{j:g_j \in S} 2^j.$$
\begin{lemma}\label{lem: non-degenerate}
Let $\delta = \min_{i \in N} \min_{S,T: v_i(S) \neq v_i(T)} |v_i(S)-v_i(T)|$ and let $\epsilon > 0$ such that $\epsilon\cdot 2^{m} < \delta$. Then the following three statements hold.
\begin{itemize}
\item For any  $i\in N$ and $S,T \subseteq M$, $v_i(S) > v_i(T)$ implies $v'_i(S) > v'_i(T)$.
\item $I_\epsilon$ is a non-degenerate instance. 
\item If $X$ is an EFX allocation for $I_\epsilon$ then $X$ is also an EFX allocation for $I$.

\end{itemize}
\end{lemma}
The proof of Lemma~\ref{lem: non-degenerate} is almost same as the proof of Lemma~1 in \cite{chaudhury2020efx}.
\begin{proof}[proof of Lemma~\ref{lem: non-degenerate}]
The first statement of the lemma holds from the following: 
\begin{align*}
v'_i(S)-v'_i(T)& = v_i(S)-v_i(T)+\epsilon\left(\sum_{g_j\in S\setminus T} 2^j- \sum_{g_j\in T\setminus S} 2^j\right) \\
& \ge \delta - \epsilon \sum_{g_j\in T\setminus S} 2^j \\
& \ge \delta-\epsilon\cdot(2^{m}-1) \\
& > 0.
\end{align*}

For the second statement of the lemma, consider any two sets $S,T \subseteq M$ such that $S\neq T$.
For any $i \in N$, if $v_i(S)\neq v_i(T)$, we have $v'_i(S)\neq v'_i(T)$ by the first statement of the lemma.
If $v_i(S)=v_i(T)$, we have $v'_i(S)-v'_i(T)=\epsilon\left(\sum_{g_j\in S\setminus T} 2^j- \sum_{g_j\in T\setminus S} 2^j\right) \neq 0$ by the fact that $S\neq T$.

For the final statement of the lemma, let us assume that $X$ is not an EFX allocation for $I$.
Then, there exist a pair of agents $i, j\in N$ and $g\in X_j$ such that $v_i(X_j \setminus g)> v_i(X_i)$.
Now, we have $v'_i(X_j \setminus g)> v'_i(X_i)$ by the first statement of the lemma.
It implies that $X$ is not an EFX allocation for $I_\epsilon$, which is a contradiction.
\end{proof}

\end{document}